\documentclass[runningheads]{llncs}
\usepackage{graphicx,amssymb,amsmath}

\usepackage{hyperref}
\usepackage{ifthen}
\usepackage{cleveref}
\usepackage{ulem}
\usepackage{algorithm2e}
\usepackage{tikz}
\usepackage{subfigure}

\newtheorem{observation}{Observation}

\title{NP-completeness Results for Graph Burning on Geometric Graphs}
\author{Arya Tanmay Gupta\inst{1}
\and Swapnil A. Lokhande\inst{2}
\and Kaushik Mondal\inst{3}}
\institute{Michigan State University, East Lansing, Michigan, USA \and Indian Institute of Information Technology Vadodara, India
\and Indian Institute of Technology Ropar, India}
\authorrunning{A. T. Gupta et al.}

\begin{document}

\maketitle

\begin{abstract}
    Graph burning runs on discrete time steps. The aim is to burn all the vertices in a given graph in the least number of time steps. This number is known to be the burning number of the graph. The spread of social influence, an alarm, or a social contagion can be modeled using graph burning. The less the burning number, the faster the spread.

    Optimal burning of general graphs is NP-Hard. There is a 3-approximation algorithm to burn general graphs where as better approximation factors are there for many sub classes. Here we study burning of grids; provide a lower bound for burning arbitrary grids and a 2-approximation algorithm for burning square grids. On the other hand, burning path forests, spider graphs, and trees with maximum degree three is already known to be NP-Complete. In this article we show burning problem to be NP-Complete on connected interval graphs, permutation graphs and several other geometric graph classes as corollaries.
\end{abstract}

\section{Introduction}

The spread of social influence in order to analyze a social network is an important topic of study \cite{Banerjee2012,Kempe2003,Kempe2005}. Kramer et al. \cite{Kramer2014} have highlighted that the underlying network plays an essential role in the spread of an emotional contagion; they have nullified the necessity of in-person interaction and non-verbal cues. With the aim to model such problems, \textit{Graph Burning} was introduced in \cite{Bonato2016}. Graph burning is also inspired by other contact processes like \textit{firefighting} \cite{Finbow2009}, \textit{graph cleaning} \cite{Alon2009}, and \textit{graph bootstrap percolation} \cite{Balogh2012}. Burning a graph can be used to model the spread of a meme, gossip, or a social contagion, influence or emotion. It can also be used to model the viral infections: the exposure to infections and proliferation of virus.

Graph burning runs on discrete time-steps (or rounds) as follows: in each time-step $t$, first (a) an arbitrary vertex is \textit{burnt} from ``outside'' (it is selected as a \textit{fire source}), and then, (b) the fire spreads to the vertices that are one hop neighbors of the already burnt vertices (burnt by round $t-1$); this process continues till all the vertices of the given graph are burned. Observe that some fire source selected at round $t$ does not spread fire to its one hop neighbors in the same round. The sequence of fire sources, selected one in each round until a graph is completely burnt, is called a \textit{burning sequence} of that graph. The minimum time steps (equivalently, number of fire sources) required to burn a graph $G$ is called the \textit{burning length} or the \textit{burning number} of $G$, and is denoted by $b(G)$. The less the value of $b(G)$, the faster it is to spread the fire, and therefore burn all the vertices in $G$. The graph burning problem is to find an optimal burning sequence for a given graph $G$. At places, we use burning problem to refer the same.

The underlying decision problem for graph burning is as follows: the given input is an arbitrary graph $G$ and an integer $k$, the task is to determine if $G$ can be burned in $k$ or less rounds. Bessy et al. \cite{Bessy2017} showed that optimal graph burning is an NP-Complete problem. They also showed that burning spider graphs, trees with maximum degree three and path forests is NP-Complete. In this article, we study the graph burning problem on \textit{interval graphs} and \textit{grids}. Interval graphs are formed from a set of closed intervals on the real line such that each interval corresponds to a vertex and the vertices corresponding to two such intervals are connected only if they overlap on the real line. Grids are formed by a set of equidistant horizontal and vertical lines intersecting at right angles such that each intersection point corresponds to a vertex and all the (induced) line segments joining those vertices are considered as edges. We prove NP-Completeness results for interval graphs. Our construction and proof technique are similar to \cite{Bessy2017}. We also provide matching bounds for burning grids.

\noindent\textbf{Our Contribution}: We provide a lower bound for the burning number of grids of arbitrary size and a 2-approximation algorithm for burning square grids (\Cref{section:burn-grids}). We prove burning connected interval graphs(\Cref{section:burn-interval-graphs}) and general permutation graphs (\Cref{section:burn-pg}) to be NP-Complete. We also report hardness results on some more graph classes (\Cref{{sec:hardgeometric}}) as corollaries of our constructions and proofs.

\section{Preliminary definitions and symbols}\label{section:PDandS}

We mention below some of the notations used in this article. Let $G$ be a graph. We denote the set of vertices in $G$ by $V(G)$. The distance between two vertices imply the number of edges contained in the shortest path between those two vertices in $G$. The radical center of a graph means the vertex from which the shortest distance to the furthest vertex is minimum. We define $\cup_{\setminus s}$ to be the \textit{left sequential union}. This operation can add a single element to a sequence, or merge two sequences. As an example, let $P = (a,b)$ be a sequence, then after executing the statement $P = P \cup_{\setminus s} (c)$, $P$ becomes $(c, a, b)$. Similarly $\cup_{s/}$ is defined as the \textit{right sequential union}. Let $Q_1$, $Q_2$ be two paths. By \textit{joining} these two paths in order $Q_1$, $Q_2$, we mean adding an edge between the last vertex of $Q_1$ and the first vertex of the $Q_2$. Let $A$ be a set of natural numbers. We denote the sum of all numbers in $A$ as $s(A)$. The largest element in $A$ is denoted by $\max(A)$.

Let $W$ be a non-empty set of vertices such that $W \subset V(G)$. The set of vertices that are at most at a distance $i$ from $W$ in $G$, including $W$ is denoted by $G.N_i[W]$. The set $W$ may be a set containing a single element.
Let $S=(x_1,x_2,...,x_k)$ be a burning sequence of $G$ of size $k$ such that, $x_i$ is chosen as the fire source in round $i$. The \textit{burning cluster} (or simply \textit{cluster}, when it is clear from the context) of a fire source $x_i$ is the set $G.N_{k-i}[x_i]$. Precisely, it is the set of vertices, to whom the fire source $x_i$ is able to spread fire to in the remaining $(k-i)$ rounds. Observe that, $x_i$ would be able to burn all of it's $(k-i)$ hop neighbors. Now it is easy to see that, if $S$ is able to burn $G$ completely, then \cref{equation:WGBC} must hold true \cite{Bessy2017}.
\begin{equation}\label{equation:WGBC}
    G.N_{k-1}[x_1] \cup G.N_{k-2}[x_2] \cup ... \cup G.N_0[x_{k}] = V(G)
\end{equation}

For the NP-Completeness proof, we reduce \textit{ distinct 3-partition problem} to our problem. The input of the  distinct 3-partition problem is a set of distinct natural numbers, $X = \{a_1, a_2, . . ., a_{3n}\}$, such that $\sum_{i=1}^{3n}a_i = nB$ where $\frac{B}{4}<a_i<\frac{B}{2}$. The task is to determine if $X$ can be partitioned into $n$ sets, each containing $3$ elements such that sum of each set equals $B$. Note that, $B$ can only be a natural number as it is a sum of 3 natural numbers. It is well known that the  distinct 3-partition problem is NP-Complete in the strong sense (see \cite{Garey1979,Bessy2017}).

\vspace{-0.3cm}
\section{Related Works}\label{section:RW}
\vspace{-0.3cm}
The burning problem was introduced by Bonato et al. (2016) \cite{Bonato2016}. This work showed that the burning number of a path or cycle of length $p$ is $\left \lceil{\sqrt{p}}\right \rceil $ along with some other properties and results. Bessy et al. \cite{Bessy2017} showed that burning a general graph is NP-Complete: they showed that burning spider graphs, trees, and path-forests are NP-Complete. A 3-approximation algorithm for burning general graphs was described in \cite{Bessy2017}. Bonato et al. \cite{Bonato2019} proposed a $2$-approximation algorithm for burning trees. A 1.5-approximation algorithm for burning path-forests was described in \cite{Bonato2019a}. A 2-approximation algorithm for burning graphs that are bounded by a diameter of constant length was described in \cite{Kamali2020}. There are works providing upper bounds on the burning number of some classes of graphs.
Authors in \cite{sandip2018} as well as \cite{Bonato2019a} showed that burning number of spider graphs of order $n$ is at most $\sqrt{n}$. Bessy et al. \cite{Bessy2018} provided a bound on the burning number of a connected graph of order $n$, and a special class of trees. Simon et al. \cite{Simon2019} presented systems that utilize burning in the spread of an alarm through a network. Simon et al.\cite{Simon2019a} provided heuristics to minimize the time steps in \textit{burning} a graph.
Kamali et al. \cite{Kamali2020} provides upper bound on burning number for the graphs with bounded path length and also for the graphs with minimum degree $\delta$. Along with this, authors in \cite{Kamali2020} (also \cite{Kare2019}) discussed bounds on the burning number of interval graphs and showed almost tight bounds. Although, they have not provided any algorithm to find an optimal burning sequence. Despite the known bounds on the burning number and the fact that most of other properties of interval graphs can be computed in polynomial time, we show that burning connected interval graphs turns out to be NP-Complete. Also we study graph burning on grids which is a graph of constant minimum degree.\\


\vspace{-0.3cm}
\section{Burning grids}\label{section:burn-grids}
\vspace{-0.3cm}
In this section we study graph burning problem on grids by providing a lower bound for grids of arbitrary size and an 2-approximation algorithm for square grids. According to \cite{Kamali2020}, the upper bound on burning number of graphs with constant minimum degree is $O(\sqrt{n})$. Here we provide a better upper bound and a matching lower bound for this specific class of graphs.

First we analyze at most how many nodes can be burnt by an arbitrary fire source inside the grid. We show an example in Figure \ref{figure:grid-example} (a). Let the maximum number of vertices that can be burned by a single fire source in $k$ rounds be denoted by $f_k$. We compute $f_k$ using the recurrence relations as follows.

\begin{center}
    $f_1=1$\\
    $f_k=4(k-1)+f_{k-1}$
\end{center}

At $k^{th}$ time step after a fire source is placed, the number of vertices which can be burned is $1+4+8+12+...+4(k-1)= 1+4(1+2+3+...+(k-1))=1+4\times\frac{k(k-1)}{2}=2k(k-1)+1$.
\vspace{-0.3cm}
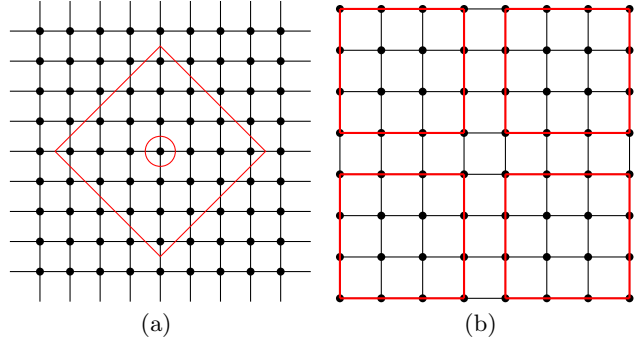
\begin{figure}[h]
    \centering
    \subfigure[]{
    \begin{tikzpicture}[scale=.4]
        \foreach \i in {0,1,...,8}{
            \foreach \j in {0,1,...,8}{
                \node [circle, fill=black, inner sep=0pt, minimum size=3pt] at (\i,\j) {};
            }
        }
        \foreach \i in {0,1,...,8}{
            \draw (\i,-1) -- (\i,9);
        }
        \foreach \i in {0,1,...,8}{
            \draw (-1,\i) -- (9,\i);
        }

        \draw[red] (.5,4) -- (4,.5) -- (7.5,4) -- (4,7.5) -- (.5,4);

        \draw[red] (4,4) circle (.5cm);
    \end{tikzpicture}
    }
    \subfigure[]{
    \begin{tikzpicture}[scale=.55]
        \foreach \i in {0,1,...,7}{
            \foreach \j in {0,1,...,7}{
                \node [circle, fill=black, inner sep=0pt, minimum size=3pt] at (\i,\j) {};
            }
        }
        \foreach \i in {0,1,...,7}{
            \draw (\i,0) -- (\i,7);
        }
        \foreach \i in {0,1,...,7}{
            \draw (0,\i) -- (7,\i);
        }
        \draw[thick,red] (0,0) -- (0,3);
        \draw[thick,red] (0,3) -- (3,3);
        \draw[thick,red] (3,3) -- (3,0);
        \draw[thick,red] (3,0) -- (0,0);

        \draw[thick,red] (0,4) -- (0,7);
        \draw[thick,red] (0,7) -- (3,7);
        \draw[thick,red] (3,7) -- (3,4);
        \draw[thick,red] (3,4) -- (0,4);

        \draw[thick,red] (4,4) -- (4,7);
        \draw[thick,red] (4,7) -- (7,7);
        \draw[thick,red] (7,7) -- (7,4);
        \draw[thick,red] (7,4) -- (4,4);

        \draw[thick,red] (4,0) -- (7,0);
        \draw[thick,red] (7,0) -- (7,3);
        \draw[thick,red] (7,3) -- (4,3);
        \draw[thick,red] (4,3) -- (4,0);
    \end{tikzpicture}
    }
    \caption{(a) On a grid, a fire source (circled) is able to burn $1$ vertex in $1$ round, $5$ vertices in $2$ rounds, $13$ vertices in $3$ rounds, $25$ vertices in $4$ rounds, and so on. (b) An $8\times 8$ grid divided into four $4\times 4$ subgrids.}
    \label{figure:grid-example}
\end{figure}

\noindent \textbf{The Lower Bound:} We prove the following lemma on the lower bound of burning number on any arbitrary grid of size $l\times b$.
\begin{lemma}
    To burn a grid of size $l\times b$ or more, we need a burning sequence containing at least $(l\times b)^\frac{1}{3}$ fire sources.
\end{lemma}
\vspace{-0.3cm}
\begin{proof}
Consider a grid of size $l\times b$, where $l$ and $b$ are any positive integers. As discussed above, $f_k=2k(k-1)+1$. So, if $i$ be the burning number of the grid, then the total number of vertices that are burned by this burning sequence will be
    $1+5+13+25+...+2i(i-1)+1$
    $=(2(0+2+6+12+20+...+(i^2-i))+i)$
    $=2((1^2-1)+(2^2-2)+(3^2-3)+...+(i^2-i))+i$
    $=2((1^2+2^2+3^2+...+i^2)-(1+2+3+...+i))+i$
    $=2\Big(\frac{i(i+1)(2i+1)}{6}-\frac{i(i+1)}{2}\Big)+i$
    $=\frac{i(i+1)(2i+1)}{3}-i(i+1)+i$
    $=\frac{i(i+1)(2i-2)}{3}+i$
    $=\frac{i(2i^2-2+3)}{3}$
    $=\frac{i(2i^2+1)}{3}$.

A burning sequence of size $i$ shall be able to burn at most $\frac{2i^3+i}{3}$ vertices on the grid. So to burn the grid in $i$ rounds, we must have,
\begin{align}\label{eqn:lowerbound}
\frac{2i^3+i}{3}\geq l\times b
\end{align}

Since $\frac{2i^3+i}{3}\leq i^3 ~~\forall i\geq 1$ and the burning sequence burns all the vertices, we have $i^3 \geq \frac{2i^3+i}{3}\geq l\times b$. This implies $i \geq (l\times b)^\frac{1}{3}$.
\qed\end{proof}
We state the following corollary.
\vspace{-0.2cm}
\begin{corollary}
To burn a square grid of size $n$, the burning number needs to be at least $n^{\frac{1}{3}}$.
\end{corollary}

Now we describe the following algorithm to burn an arbitrary grid $G$ of size $l\times b$. We further show that this is a 2-approximation algorithm for burning any $l\times l$ square grid with $l\geq 403$.

\vspace{0.2cm}
\noindent \textbf{The Algorithm:} Divide $G$ into subgrids (see \Cref{figure:grid-example} (b) for example) of dimensions $l^{\frac{2}{3}}\times b^{\frac{2}{3}}$. Represent the resultant subgrids by $g_1,g_2,...,g_k$ where $k$ is the count of subgrids obtained. Let $S$ be the sequence of fire sources, initially empty. For $1\leq i \leq k$, put the radical center $x_i$ of subgrid $g_i$ as the $i$-th fire source in $S$. If $G$ is not completely burnt by those $k$ fire sources, then in each step $i\geq k+1$, continue putting unburnt vertices from $G$ in $S$ until $G$ is completely burnt.

\begin{theorem}\label{theorem:burn-grid-2-approx}
    Our algorithm is able to burn a grid $G$ within an approximation factor of $2$ if $G$ is a $l\times l$ square grid with $l\geq 403$.
\end{theorem}
\vspace{-0.3cm}
\begin{proof}
    The algorithm divides the grid in to at most $(l^{\frac{1}{3}}+1)\times (l^{\frac{1}{3}}+1)$ subgrids.
    In each round $i$, $x_i$ is set to be the $i$-th fire source. As the fire source is placed at the radical center of a subgrid, and the radius
    of the subgrid is $l^{\frac{2}{3}}$, so it takes $l^{\frac{2}{3}}$ rounds to burn the corresponding subgrid. As the last fire source may take up to $l^{\frac{2}{3}}$  rounds to burn the respective subgrid, our algorithm takes a total of at most $\lceil{(l^{\frac{1}{3}}+1)\times (l^{\frac{1}{3}}+1)+l^{\frac{2}{3}}}\rceil$ rounds i.e., $\lceil{2 l^{\frac{2}{3}}+ 2 l^{\frac{1}{3}} +1}\rceil$ rounds to burn $G$ completely.
    Next we see what can not be burnt using half of the rounds that our algorithm takes in worst case.

    It is impossible to burn any $l\times l$ square grid with $l\geq 403$ in less than $\frac{\lceil{2 l^{\frac{2}{3}}+ 2 l^{\frac{1}{3}} +1}\rceil}{2}$ rounds as $i=\frac{\lceil {2 l^{\frac{2}{3}}+ 2 l^{\frac{1}{3}} +1}\rceil}{2}$ value does not satisfy Equation \ref{eqn:lowerbound} for $b=l$ where $l\geq 403$. Remember, to burn $G$ completely in $i$ rounds, Equation \ref{eqn:lowerbound} must be satisfied. Hence the proof.
%
\qed\end{proof}


\vspace{-0.3cm}
\section{Burning interval graphs}\label{section:burn-interval-graphs}
\vspace{-0.3cm}
We show that burning connected interval graphs is NP-Complete by giving a reduction from the distinct 3-partition problem. We construct interval graph from any given instance of the  distinct 3-partition problem. We do so by replacing each spider structure by a ``comb structure'' in the construction of the NP-Completeness proof for burning trees in \cite{Bessy2017}, which we elaborate later in this section. But before going to that, we have the following discussion that tries to relate burning an interval graph to burning a path.

Bonato et al. (2016) \cite{Bonato2016} proved that a path or a cycle of $n$ vertices can be \textit{burned} in $\lceil \sqrt{n} \rceil$ steps. Note the following observation from the above fact.
\begin{observation}\label{observation:path}
The burning clusters of each of the $n$ fire sources of any optimal burning sequence of a path of $n^2$ vertices are pairwise disjoint.
\end{observation}
\vspace{-0.2cm}


We provide an example of burning a path of size nine as shown in the \Cref{figure:pb}. The vertex $v_3$ is chosen as fire source $x_1$ in time $t=1$ and $x_1$ is burned at this step. In time $t=2$, $v_7$ is chosen as the next fire source $x_2$ and it is burned in this step. Along with this, the one hop neighbors $v_2$, $v_4$ of the already burnt (by step 1) vertex $v_3$ also are burned in this step. In time $t=3$, $v_9$ is selected as the third fire source and subsequently is burned in this step. Also the one hop neighbors $v_1$, $v_5$, $v_6$ and $v_8$ are burned in this step by the spread of fire from the already burnt vertices (by step 2).

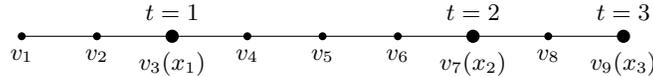
\begin{figure}[ht!]
    \centering
    \begin{tikzpicture}
        \node [circle, fill=black, inner sep=0pt, minimum size=3pt, label=below:$v_1$] (A) at (0,0) {};
        \node [circle, fill=black, inner sep=0pt, minimum size=3pt, label=below:$v_2$] (B) at (1,0) {};
        \node [circle, fill=black, inner sep=0pt, minimum size=5pt, label=below:{$v_3(x_1)$}, label=above:{$t=1$}] (C) at (2,0) {};
        \node [circle, fill=black, inner sep=0pt, minimum size=3pt, label=below:$v_4$] (D) at (3,0) {};
        \node [circle, fill=black, inner sep=0pt, minimum size=3pt, label=below:$v_5$] (E) at (4,0) {};
        \node [circle, fill=black, inner sep=0pt, minimum size=3pt, label=below:$v_6$] (F) at (5,0) {};
        \node [circle, fill=black, inner sep=0pt, minimum size=5pt, label=below:{$v_7(x_2)$}, label=above:{$t=2$}] (G) at (6,0) {};
        \node [circle, fill=black, inner sep=0pt, minimum size=3pt, label=below:$v_8$] (H) at (7,0) {};
        \node [circle, fill=black, inner sep=0pt, minimum size=5pt, label=below:{$v_9(x_3)$}, label=above:{$t=3$}] (I) at (8,0) {};

        \draw (0,0) -- (8,0);
    \end{tikzpicture}
    \caption{Path of length nine is burnt in three steps.}
    \label{figure:pb}
\end{figure}

We would like to recall another result from \cite{Kamali2020,Kare2019} on the bounds on burning  number of interval graphs as the following observation.
\begin{observation}\label{observation:bound}
Let $L$ be a maximum length path among the all pair shortest paths of an interval graph $G$. Then $b(L) \leq b(G) \leq b(L)+1$.
\end{observation}
Also note that, finding such $L$ is easy to do in polynomial time. We can simply compute all pair shortest path and choose the maximum length path among all.
Then we can see from the proof of Observation \ref{observation:bound} that burning an interval graph in $(b(L)+1)$ rounds, i.e., at most in $(b(G)+1)$ rounds is trivial. We study whether finding a burning sequence of length $b(G)$ is possible in polynomial time, especially if $b(G)=b(L)$ for an interval graph $G$. We show that determining whether $b(G)=b(L)$ is NP-Complete.

\vspace{0.2cm}
\noindent \textbf{General idea:} First we provide a general idea behind our approach. We prove the Np-Completeness of burning interval graphs by giving a reduction from the distinct 3-partition problem. We construct interval graph from any given instance of the  distinct 3-partition problem. We show that burning this interval graph is possible optimally in polynomial time if and only if one can solve the  distinct 3-partition problem. While describing the idea, we refer to few notations here which are defined in Section \ref{sec:IGconstruct}.

We start with any input $X$ of the  distinct 3-partition problem.
First we construct another set $X'$ from $X$ such that all the elements of $X'$ are odd. The reason behind moving to $X'$ is, we aim to use the fact that the sizes of the burning clusters of the fire sources on a path are all odds if the length of the path is a perfect square. First  we construct a path $P_I$ of length $(2m+1)^2$ (where $m = \max (X)$) by combining few subpaths of shorter lengths. Note that $b(P_I)=2m+1$. Then we add few vertices and corresponding edges to some of the subpaths $T_j$ of $P_I$ in such a way that it remains an interval graph. We call it $IG(X)$ (Section \ref{sec:IGconstruct}). The optimal burning number $b(IG(X))$ takes the value $2m+1$ whenever $X'$(and eventually $X$) can be partitioned according to the  distinct 3-partition problem (Lemma \ref{lemma:BNb(G)IG}). So we keep the burning number of the path $P_I$ and the interval graph $IG(X)$ same.

Additional vertices and edges are added to the sub paths $T_j$ to form structures $T^c_j$ (refer Figure \ref{figure:TStructureIG}) in such a way that to burn $IG(X)$ optimally, one must have to burn each $T^c_j$ only with one fire source (Lemma \ref{lemma:notopt}). Not only that, one must have to put that fire source on $T^c_j$ in a particular round depending on the length of the subpath $T_j$ (Lemma \ref{lemma:AFSOnr_iIG}).
With the help of these results and another couple of results, we finally show that, to burn this interval graph optimally in $b(IG(X))$ steps, one needs to solve the  distinct 3-partition problem on the input $X$. This makes our problem an NP-Complete problem (Section \ref{sec:IGoptseq} Theorem \ref{theorem:BIGNPCIG}).

\vspace{-0.3cm}
\subsection{Interval graph construction}\label{sec:IGconstruct}

Let n be a natural number. Let $X=\{a_1, a_2, \cdots, a_{3n}\}$ be an input to a  distinct 3-partition problem. So, $n=\frac{|X|}{3}$ and $B = \frac{s(X)}{n}$.
Let $m = \max (X)$, and $k=m-3n$. Let $F_m$ be the set of first $m$ natural numbers, $F_m = \{1,2,3,...,m\}$. Also let $F^{\prime}_m$ be the set of first $m$ odd numbers, $F^\prime_m = \{2\ f_i-1: f_i \in F_m\} = \{1, 3, 5, . . ., 2m-1\}$. Let $X^\prime = \{2\ a_i-1:a_i \in X\}$, $B^\prime = \frac{s(X^\prime)}{n}$. Observe that $s(X^\prime) = \sum_{i=1}^{3n} 2\ a_i -1 = 2nB-3n$, so $B^\prime = 2B-3$. It is easy to observe that any solution of $X$ gives a solution of $X'$ and vice versa. Let $Y=F_m^\prime\setminus X^\prime$.

Let there be $n$ paths $Q_1,Q_2...,Q_n$, each of order $B^\prime$. Consider $k$ paths $Q_1^\prime$, $Q_2^\prime$, $...$, $Q_k^\prime$ such that each $Q_j^\prime$ $(\forall\ 1\leq j\leq k)$ is of order of $j^{th}$ largest number in $Y$, where $k=|Y|$. Clearly the total number of vertices in $Q_1,Q_2...,Q_n,Q_1^\prime,Q_2^\prime,...,Q_k^\prime$ is $m^2$, i.e., equals $s(F^\prime_m)$.
Consider another $m+1$ paths $T_1,T_2,...,T_{m+1}$ such that each $T_j$ $(\forall\ 1\leq j\leq m+1)$ is of order of $2(2m+1-j)+1$. Total number of vertices in
$T_1,T_2,...,T_{m+1}$ is $\sum_{j=1}^{m+1} \left(2(2m+1-j)+1\right)= (3m^2 +4m +1)$.
We join these paths in the following order to form a larger path:\\
$Q_1$, $T_1$, $Q_2$, $T_2$, $...$, $Q_n$, $T_n$, $Q_1^\prime,T_{n+1}$, $Q_2^\prime$, $T_{n+2}$, $...$, $Q_k^\prime$, $T_{n+k}$, $T_{n+k+1}$, $...$, $T_{m+1}$.
We denote this path as $P_I$. The total number of vertices in $P_I$ is $m^2 + 3m^2 +4m +1 = (2m+1)^2$. Hence $b(P_I)=(2m+1)$.

Now we add few more vertices to $P_I$ in such a way that it remains an interval graph and the optimal burning number of the graph remains same as $b(P_I)$. We add a distinct vertex connected to each vertex from $2nd$ to $2nd$-last vertices of $T_j$, $\forall\ 1\leq j\leq m+1$ (\Cref{figure:TStructureIG} illustrates an example $T_j$ along with the added vertices and edges (vertically upwards w.r.t. $T_j$). This forms a kind of comb structure; we call it $T^c_j$).
Let this graph be called $IG(X)$. Now we calculate total number of vertices in $IG(X)$. Number of vertices added to each $T_j$ is $\left(2(2m+1-j)+1\right)-2 = \left(4m+1-j\right)$. Hence total number of vertices added to $P_I$ is $\sum_{j=1}^{m+1}\left(4m+1-j\right) = (3m^2+2m-1)$. So, total number of vertices in $IG(X)$ is $(2m+1)^2+ (3m^2+2m-1) = (7m^2+6m)$. One such example of $IG(X)$ is shown in Figure \ref{figure:IGNPCEIG} corresponding to the numerical example given in \Cref{subsection:ExConsIG}.

Observe that $P_I$ is a diameter of $IG(X)$ and there is no cycle in $IG(X)$. Also, all the vertices which are not in $P_I$ are connected to some vertex of $P_I$ by an edge. Hence $IG(X)$ is a valid interval graph.

\newcounter{c}
\newcounter{d}
\newcounter{i}
\begin{figure}[h]
    \centering
    \begin{tikzpicture}[scale=.7]
        \setcounter{c}{0}
        \setcounter{d}{-1}
        \setcounter{i}{1}
        \loop
            \ifthenelse{\value{c}=5}{
                \node [circle, fill=black, inner sep=0pt, minimum size=3pt, label=below:{$x_1$}]      (A\thec) at (\value{c}*.3,0) {};
            }{
                \ifthenelse{\value{c}=11}{
                   \node [circle, fill=black, inner sep=0pt, minimum size=3pt, label=below:{$v\ v_r$}] (A\thec) at (\value{c}*.3,0) {};
                }{
                    \ifthenelse{\value{c}=20}{
                        \node [circle, fill=black, inner sep=0pt, minimum size=3pt, label=below:{$x_2$}] (A\thec) at (\value{c}*.3,0) {};
                    }{
                        \node [circle, fill=black, inner sep=0pt, minimum size=3pt] (A\thec) at (\value{c}*.3,0) {};
                    }
                }
            }

            \ifthenelse{\value{c}>0 \AND \value{c}<32}{
                \ifthenelse{\value{c}=1}{
                    \node [circle, fill=black, inner sep=0pt, minimum size=3pt, label=left:{$u_j^1$}] (B\thec) at (\value{c}*.3,1) {};

                    \draw[blue] (A\thec) -- (B\thec);
                }{
                    \ifthenelse{\value{c}=31}{
                        \node [circle, fill=black, inner sep=0pt, minimum size=3pt, label=right:{$u_j^{|T_j|-2}$}] (B\thec) at (\value{c}*.3,1) {};

                        \draw[blue] (A\thec) -- (B\thec);
                    }{
                        \ifthenelse{\value{c}=10}{
                            \node [circle, fill=black, inner sep=0pt, minimum size=3pt, label=above:{$v_t$}] (B\thec) at (\value{c}*.3,1) {};

                            \draw[blue] (A\thec) -- (B\thec);
                        }{
                            \node [circle, fill=black, inner sep=0pt, minimum size=3pt] (B\thec) at (\value{c}*.3,1) {};

                            \draw[blue] (A\thec) -- (B\thec);
                        }
                    }
                }
            }{}

            \ifthenelse{\value{c}>0}{
                \draw (A\thec) -- (A\thed);
            }{}

            \ifthenelse{\value{c}=0}{
                \draw[red,dashed] (-1,0) -- (A\thec);
            }{}
            \ifthenelse{\value{c}=32}{
                \draw[red,dashed] (A\thec) -- (\value{c}*.3+1,0);
            }{}

            \stepcounter{c}
            \stepcounter{d}
            \stepcounter{i}
            \ifnum \value{c}<33
            \repeat

            \draw[<-] (1.9,-.3) -- (2.6,-.3);
            \draw[->] (3.7,-.3) -- (5.6,-.3);
    \end{tikzpicture}
    \vspace{-0.4cm}
    \caption{Structure of a $T_j$ with 33 vertices, along with the extra vertices connected to it. The dashed line represents the fact that other subpaths may be connected to a $T_j$ on either or both ends.}
    \label{figure:TStructureIG}
\end{figure}
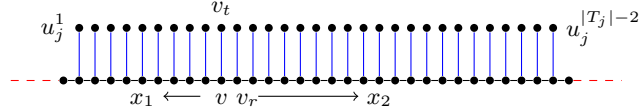

Let $u^1_j$ be the vertex connected to the $2nd$ vertex of each $T_j$ and $u^{|T_j|-2}_j$ be the vertex connected to its $2nd$-last vertex of $T_j$, where $|T_j|$ stands for the number of vertices in the subpath $T_j$. Let $A^{T_j}=\{u^1_j, u^2_j, ..., u^{|T_j|-2}_j\}$ be the set of all $|T_j|-2$ additional vertices corresponding to $T_j$. Now we mention an important observation regarding burning $T^c_j$.

\subsection{Example construction}\label{subsection:ExConsIG}

In this section, we show the construction of $IG(X)$ from a particular input set $X$.
Let $X=\{10,11,12,14,15,16\}$. Then $n=2,\ m = 16,\ B = 39,$ and $k=10$. Also $F_m=\{1,2,...,16\}$ and $F_m^\prime=\{1,3,...,31\}$. Further, $X^\prime = \{19,21,23,27,29,31\}$, $B^\prime = 75=2B-3$ and $Y=\{1,3,5,7,9,11,13,15,17,25\}$.
Observe that $Q_1$ and $Q_2$ are paths of size $75$, and each $Q_1^\prime, Q_2^\prime,...,Q_{k}^\prime$ are paths of order of $25$, $17$, $15$, $13$, $11$, $9$, $7$, $5$, $3$, $1$ respectively. $T_1,T_2,T_3,...T_{m+1}$ are of order of $65,63, 61...,33$ respectively.
We add a vertex connected to each vertex from $2nd$ to $2nd$-last vertices of $T_j (1\leq j\leq m+1)$.
Observe that this is a valid interval graph. The constructed example $IG(X)$ is shown in \Cref{figure:IGNPCEIG}.

\newcounter{n}
\newcounter{r}
\begin{figure}
    \centering
    \begin{tikzpicture}[scale=.6]
        \setcounter{c}{0}
        \setcounter{r}{0}
        \loop
            \node [circle, fill=black, inner sep=0pt, minimum size=3pt] (A) at (\value{r}*3.75,-\value{c}*2) {};
            \ifnum \value{r}>0
                \draw[red] (B) -- (A);
            \fi
            \node [circle, fill=black, inner sep=0pt, minimum size=3pt] (B) at (\value{r}*3.75+2.75,-\value{c}*2) {};

            \draw (A)--(B);

            \ifthenelse{\isodd{\value{r}}}{
                \node [circle, fill=black, inner sep=0pt, minimum size=3pt] (P1) at (\value{r}*3.75+.2,-\value{c}*2) {};
                \node [circle, fill=black, inner sep=0pt, minimum size=3pt] (P2) at (\value{r}*3.75+2.75-.2,-\value{c}*2) {};

                \node [circle, fill=black, inner sep=0pt, minimum size=3pt] (U1) at (\value{r}*3.75+.2,-\value{c}*2+.75) {};
                \node [circle, fill=black, inner sep=0pt, minimum size=3pt] (U2) at (\value{r}*3.75+2.75-.2,-\value{c}*2+.75) {};

                \draw[blue] (P1) -- (U1);
                \draw[blue] (P2) -- (U2);

                \node [circle, fill=black, inner sep=0pt, minimum size=3pt] (P3) at (\value{r}*3.75+.4,-\value{c}*2) {};
                \node [circle, fill=black, inner sep=0pt, minimum size=3pt] (P4) at (\value{r}*3.75+2.75-.4,-\value{c}*2) {};

                \node [circle, fill=black, inner sep=0pt, minimum size=3pt] (U3) at (\value{r}*3.75+.4,-\value{c}*2+.75) {};
                \node [circle, fill=black, inner sep=0pt, minimum size=3pt] (U4) at (\value{r}*3.75+2.75-.4,-\value{c}*2+.75) {};

                \draw[blue] (P1) -- (U1);
                \draw[blue] (P2) -- (U2);

                \node [circle, fill=black, inner sep=0pt, minimum size=3pt] (P5) at (\value{r}*3.75+.6,-\value{c}*2) {};
                \node [circle, fill=black, inner sep=0pt, minimum size=3pt] (P6) at (\value{r}*3.75+2.75-.6,-\value{c}*2) {};

                \node [circle, fill=black, inner sep=0pt, minimum size=3pt] (U5) at (\value{r}*3.75+.6,-\value{c}*2+.75) {};
                \node [circle, fill=black, inner sep=0pt, minimum size=3pt] (U6) at (\value{r}*3.75+2.75-.6,-\value{c}*2+.75) {};

                \draw[blue] (P3) -- (U3);
                \draw[blue] (P4) -- (U4);
                \draw[blue] (P5) -- (U5);
                \draw[blue] (P6) -- (U6);

                \node [circle, fill=black, inner sep=0pt, minimum size=2pt] at (\value{r}*3.75+2.75/2,-\value{c}*2+.375) {};
                \node [circle, fill=black, inner sep=0pt, minimum size=2pt] at (\value{r}*3.75+2.75/2-.2,-\value{c}*2+.375) {};
                \node [circle, fill=black, inner sep=0pt, minimum size=2pt] at (\value{r}*3.75+2.75/2+.2,-\value{c}*2+.375) {};
            }{}

            \stepcounter{r}
            \ifnum\value{r}<4
            \repeat

        \node at (1.375,-.5) {$Q_1$}; \node at (5.125,-.5) {$T_1$};
        \node at (8.625,-.5) {$Q_2$}; \node at (13,-.5) {$T_2$};

        \stepcounter{c}
        \setcounter{r}{5}
        \loop
            \node [circle, fill=black, inner sep=0pt, minimum size=3pt] (A\ther) at (\value{r}*2.5,-\value{c}*2) {};
            \node [circle, fill=black, inner sep=0pt, minimum size=3pt] (B\ther) at (\value{r}*2.5+1.5,-\value{c}*2) {};
            \ifnum\value{r}=5
                \draw[red] (B\ther) -- (B);
            \fi
            \ifnum\value{r}<5
                \setcounter{n}{\value{r}}
                \stepcounter{n}
                \draw[red] (B\ther) -- (A\then);
            \fi

            \draw (A\ther)--(B\ther);

            \ifthenelse{\isodd{\value{r}}}{}{
                \node [circle, fill=black, inner sep=0pt, minimum size=3pt] (P1) at (\value{r}*2.5+.2,-\value{c}*2) {};
                \node [circle, fill=black, inner sep=0pt, minimum size=3pt] (P2) at (\value{r}*2.5+1.5-.2,-\value{c}*2) {};

                \node [circle, fill=black, inner sep=0pt, minimum size=3pt] (U1) at (\value{r}*2.5+.2,-\value{c}*2+.75) {};
                \node [circle, fill=black, inner sep=0pt, minimum size=3pt] (U2) at (\value{r}*2.5+1.5-.2,-\value{c}*2+.75) {};

                \draw[blue] (P1) -- (U1);
                \draw[blue] (P2) -- (U2);

                \node [circle, fill=black, inner sep=0pt, minimum size=3pt] (P3) at (\value{r}*2.5+.4,-\value{c}*2) {};
                \node [circle, fill=black, inner sep=0pt, minimum size=3pt] (P4) at (\value{r}*2.5+1.5-.4,-\value{c}*2) {};

                \node [circle, fill=black, inner sep=0pt, minimum size=3pt] (U3) at (\value{r}*2.5+.4,-\value{c}*2+.75) {};
                \node [circle, fill=black, inner sep=0pt, minimum size=3pt] (U4) at (\value{r}*2.5+1.5-.4,-\value{c}*2+.75) {};

                \draw[blue] (P3) -- (U3);
                \draw[blue] (P4) -- (U4);

                \node [circle, fill=black, inner sep=0pt, minimum size=3pt] (P3) at (\value{r}*2.5+.4,-\value{c}*2) {};
                \node [circle, fill=black, inner sep=0pt, minimum size=3pt] (P4) at (\value{r}*2.5+1.5-.4,-\value{c}*2) {};

                \node [circle, fill=black, inner sep=0pt, minimum size=3pt] (U3) at (\value{r}*2.5+.4,-\value{c}*2+.75) {};
                \node [circle, fill=black, inner sep=0pt, minimum size=3pt] (U4) at (\value{r}*2.5+1.5-.4,-\value{c}*2+.75) {};

                \draw[blue] (P3) -- (U3);
                \draw[blue] (P4) -- (U4);

                \node [circle, fill=black, inner sep=0pt, minimum size=2pt] (U4) at (\value{r}*2.5+1.5/2,-\value{c}*2+.375) {};
                \node [circle, fill=black, inner sep=0pt, minimum size=2pt] (U4) at (\value{r}*2.5+1.5/2+.15,-\value{c}*2+.375) {};
                \node [circle, fill=black, inner sep=0pt, minimum size=2pt] (U4) at (\value{r}*2.5+1.5/2-.15,-\value{c}*2+.375) {};
            }

            \addtocounter{r}{-1}
            \ifnum\value{r}>-1
            \repeat

        \node at (.75,-2.5) {$T_5$}; \node at (3.25,-2.5) {$Q_3^\prime$};
        \node at (5.75,-2.5) {$T_4$}; \node at (8.25,-2.5) {$Q_2^\prime$};
        \node at (10.75,-2.5) {$T_3$}; \node at (13.25,-2.5) {$Q_1^\prime$};

        \stepcounter{c}
        \setcounter{r}{0}
        \loop
            \node [circle, fill=black, inner sep=0pt, minimum size=3pt] (A) at (\value{r}*2.5,-\value{c}*2) {};
            \ifnum\value{r}=0
                \draw[red] (A) -- (A0);
            \fi
            \ifnum \value{r}>0
                \draw[red] (B) -- (A);
            \fi
            \node [circle, fill=black, inner sep=0pt, minimum size=3pt] (B) at (\value{r}*2.5+1.5,-\value{c}*2) {};

            \draw (A)--(B);
            \ifthenelse{\isodd{\value{r}}}{
                \node [circle, fill=black, inner sep=0pt, minimum size=3pt] (P1) at (\value{r}*2.5+.2,-\value{c}*2) {};
                \node [circle, fill=black, inner sep=0pt, minimum size=3pt] (P2) at (\value{r}*2.5+1.5-.2,-\value{c}*2) {};

                \node [circle, fill=black, inner sep=0pt, minimum size=3pt] (U1) at (\value{r}*2.5+.2,-\value{c}*2+.75) {};
                \node [circle, fill=black, inner sep=0pt, minimum size=3pt] (U2) at (\value{r}*2.5+1.5-.2,-\value{c}*2+.75) {};

                \draw[blue] (P1) -- (U1);
                \draw[blue] (P2) -- (U2);

                \node [circle, fill=black, inner sep=0pt, minimum size=3pt] (P3) at (\value{r}*2.5+.4,-\value{c}*2) {};
                \node [circle, fill=black, inner sep=0pt, minimum size=3pt] (P4) at (\value{r}*2.5+1.5-.4,-\value{c}*2) {};

                \node [circle, fill=black, inner sep=0pt, minimum size=3pt] (U3) at (\value{r}*2.5+.4,-\value{c}*2+.75) {};
                \node [circle, fill=black, inner sep=0pt, minimum size=3pt] (U4) at (\value{r}*2.5+1.5-.4,-\value{c}*2+.75) {};

                \draw[blue] (P3) -- (U3);
                \draw[blue] (P4) -- (U4);

                \node [circle, fill=black, inner sep=0pt, minimum size=2pt] (U4) at (\value{r}*2.5+1.5/2,-\value{c}*2+.375) {};
                \node [circle, fill=black, inner sep=0pt, minimum size=2pt] (U4) at (\value{r}*2.5+1.5/2+.15,-\value{c}*2+.375) {};
                \node [circle, fill=black, inner sep=0pt, minimum size=2pt] (U4) at (\value{r}*2.5+1.5/2-.15,-\value{c}*2+.375) {};
            }{}

            \stepcounter{r}
            \ifnum\value{r}<6
            \repeat

        \node at (.75,-4.5) {$Q_4^\prime$}; \node at (3.25,-4.5) {$T_6$};
        \node at (5.75,-4.5) {$Q_5^\prime$}; \node at (8.25,-4.5) {$T_7$};
        \node at (10.75,-4.5) {$Q_6^\prime$}; \node at (13.25,-4.5) {$T_8$};

        \stepcounter{c}
        \setcounter{r}{5}
        \loop
            \node [circle, fill=black, inner sep=0pt, minimum size=3pt] (A\ther) at (\value{r}*2.5,-\value{c}*2) {};
            \node [circle, fill=black, inner sep=0pt, minimum size=3pt] (B\ther) at (\value{r}*2.5+1.5,-\value{c}*2) {};
            \ifnum\value{r}=5
                \draw[red] (B\ther) -- (B);
            \fi
            \ifnum\value{r}<5
                \setcounter{n}{\value{r}}
                \stepcounter{n}
                \draw[red] (B\ther) -- (A\then);
            \fi

            \draw (A\ther)--(B\ther);

            \ifthenelse{\isodd{\value{r}}}{}{
                \node [circle, fill=black, inner sep=0pt, minimum size=3pt] (P1) at (\value{r}*2.5+.2,-\value{c}*2) {};
                \node [circle, fill=black, inner sep=0pt, minimum size=3pt] (P2) at (\value{r}*2.5+1.5-.2,-\value{c}*2) {};

                \node [circle, fill=black, inner sep=0pt, minimum size=3pt] (U1) at (\value{r}*2.5+.2,-\value{c}*2+.75) {};
                \node [circle, fill=black, inner sep=0pt, minimum size=3pt] (U2) at (\value{r}*2.5+1.5-.2,-\value{c}*2+.75) {};
                \draw[blue] (P1) -- (U1);
                \draw[blue] (P2) -- (U2);\node [circle, fill=black, inner sep=0pt, minimum size=3pt] (P3) at (\value{r}*2.5+.4,-\value{c}*2) {};
                \node [circle, fill=black, inner sep=0pt, minimum size=3pt] (P4) at (\value{r}*2.5+1.5-.4,-\value{c}*2) {};

                \node [circle, fill=black, inner sep=0pt, minimum size=3pt] (U3) at (\value{r}*2.5+.4,-\value{c}*2+.75) {};
                \node [circle, fill=black, inner sep=0pt, minimum size=3pt] (U4) at (\value{r}*2.5+1.5-.4,-\value{c}*2+.75) {};

                \draw[blue] (P3) -- (U3);
                \draw[blue] (P4) -- (U4);

                \node [circle, fill=black, inner sep=0pt, minimum size=2pt] (U4) at (\value{r}*2.5+1.5/2,-\value{c}*2+.375) {};
                \node [circle, fill=black, inner sep=0pt, minimum size=2pt] (U4) at (\value{r}*2.5+1.5/2+.15,-\value{c}*2+.375) {};
                \node [circle, fill=black, inner sep=0pt, minimum size=2pt] (U4) at (\value{r}*2.5+1.5/2-.15,-\value{c}*2+.375) {};
            }

            \addtocounter{r}{-1}
            \ifnum\value{r}>-1
            \repeat

        \node at (.75,-6.5) {$T_{11}$}; \node at (3.25,-6.5) {$Q_9^\prime$};
        \node at (5.75,-6.5) {$T_{10}$}; \node at (8.25,-6.5) {$Q_8^\prime$};
        \node at (10.75,-6.5) {$T_9$}; \node at (13.25,-6.5) {$Q_7^\prime$};

        \stepcounter{c}
        \setcounter{r}{0}
        \loop

            \ifnum\value{r}=0
                \node [circle, fill=black, inner sep=0pt, minimum size=3pt] (a) at (\value{r}*2.5,-\value{c}*2) {};
                \draw[red] (a) -- (A0);
            \fi
            \ifnum\value{r}>0
                \node [circle, fill=black, inner sep=0pt, minimum size=3pt] (A) at (\value{r}*2.5,-\value{c}*2) {};
                \ifnum \value{r}=1
                    \draw[red] (A) -- (a);
                \fi
                \ifnum \value{r}>1
                    \draw[red] (B) -- (A);
                \fi
                \node [circle, fill=black, inner sep=0pt, minimum size=3pt] (B) at (\value{r}*2.5+1.5,-\value{c}*2) {};

                \draw (A)--(B);

                \node [circle, fill=black, inner sep=0pt, minimum size=3pt] (P1) at (\value{r}*2.5+.2,-\value{c}*2) {};
                \node [circle, fill=black, inner sep=0pt, minimum size=3pt] (P2) at (\value{r}*2.5+1.5-.2,-\value{c}*2) {};

                \node [circle, fill=black, inner sep=0pt, minimum size=3pt] (U1) at (\value{r}*2.5+.2,-\value{c}*2+.75) {};
                \node [circle, fill=black, inner sep=0pt, minimum size=3pt] (U2) at (\value{r}*2.5+1.5-.2,-\value{c}*2+.75) {};\node [circle, fill=black, inner sep=0pt, minimum size=3pt] (P3) at (\value{r}*2.5+.4,-\value{c}*2) {};
                \node [circle, fill=black, inner sep=0pt, minimum size=3pt] (P4) at (\value{r}*2.5+1.5-.4,-\value{c}*2) {};

                \node [circle, fill=black, inner sep=0pt, minimum size=3pt] (U3) at (\value{r}*2.5+.4,-\value{c}*2+.75) {};
                \node [circle, fill=black, inner sep=0pt, minimum size=3pt] (U4) at (\value{r}*2.5+1.5-.4,-\value{c}*2+.75) {};

                \draw[blue] (P3) -- (U3);
                \draw[blue] (P4) -- (U4);

                \node [circle, fill=black, inner sep=0pt, minimum size=2pt] (U4) at (\value{r}*2.5+1.5/2,-\value{c}*2+.375) {};
                \node [circle, fill=black, inner sep=0pt, minimum size=2pt] (U4) at (\value{r}*2.5+1.5/2+.15,-\value{c}*2+.375) {};
                \node [circle, fill=black, inner sep=0pt, minimum size=2pt] (U4) at (\value{r}*2.5+1.5/2-.15,-\value{c}*2+.375) {};
            \fi

            \draw[blue] (P1) -- (U1);
            \draw[blue] (P2) -- (U2);

            \stepcounter{r}
            \ifnum\value{r}<6
            \repeat

        \node at (.2,-8.5) {$Q_{10}^\prime$}; \node at (3.25,-8.5) {$T_{12}$};
        \node at (5.75,-8.5) {$T_{13}$}; \node at (8.25,-8.5) {$T_{14}$};
        \node at (10.75,-8.5) {$T_{15}$}; \node at (13.25,-8.5) {$T_{16}$};

        \stepcounter{c}
        \setcounter{r}{5}
        \loop
            \node [circle, fill=black, inner sep=0pt, minimum size=3pt] (A\ther) at (\value{r}*2.5,-\value{c}*2) {};
            \node [circle, fill=black, inner sep=0pt, minimum size=3pt] (B\ther) at (\value{r}*2.5+1.5,-\value{c}*2) {};
            \ifnum\value{r}=5
                \draw[red] (B\ther) -- (B);
            \fi
            \ifnum\value{r}<5
                \setcounter{n}{\value{r}}
                \stepcounter{n}
                \draw[red] (B\ther) -- (A\then);
            \fi

            \draw (A\ther)--(B\ther);

            \node [circle, fill=black, inner sep=0pt, minimum size=3pt] (P1) at (\value{r}*2.5+.2,-\value{c}*2) {};
            \node [circle, fill=black, inner sep=0pt, minimum size=3pt] (P2) at (\value{r}*2.5+1.5-.2,-\value{c}*2) {};

            \node [circle, fill=black, inner sep=0pt, minimum size=3pt] (U1) at (\value{r}*2.5+.2,-\value{c}*2+.75) {};
            \node [circle, fill=black, inner sep=0pt, minimum size=3pt] (U2) at (\value{r}*2.5+1.5-.2,-\value{c}*2+.75) {};

            \draw[blue] (P1) -- (U1);
            \draw[blue] (P2) -- (U2);

            \node [circle, fill=black, inner sep=0pt, minimum size=3pt] (P3) at (\value{r}*2.5+.4,-\value{c}*2) {};
            \node [circle, fill=black, inner sep=0pt, minimum size=3pt] (P4) at (\value{r}*2.5+1.5-.4,-\value{c}*2) {};

            \node [circle, fill=black, inner sep=0pt, minimum size=3pt] (U3) at (\value{r}*2.5+.4,-\value{c}*2+.75) {};
            \node [circle, fill=black, inner sep=0pt, minimum size=3pt] (U4) at (\value{r}*2.5+1.5-.4,-\value{c}*2+.75) {};

            \draw[blue] (P3) -- (U3);
            \draw[blue] (P4) -- (U4);

            \node [circle, fill=black, inner sep=0pt, minimum size=2pt] (U4) at (\value{r}*2.5+1.5/2,-\value{c}*2+.375) {};
            \node [circle, fill=black, inner sep=0pt, minimum size=2pt] (U4) at (\value{r}*2.5+1.5/2+.15,-\value{c}*2+.375) {};
            \node [circle, fill=black, inner sep=0pt, minimum size=2pt] (U4) at (\value{r}*2.5+1.5/2-.15,-\value{c}*2+.375) {};

            \addtocounter{r}{-1}
            \ifnum\value{r}>4
            \repeat

        \node at (13.25,-10.5) {$T_{17}$};
    \end{tikzpicture}
    \vspace{-0.5cm}
    \caption{Construction of an example $IG(X)$.}
    \label{figure:IGNPCEIG}
\end{figure}
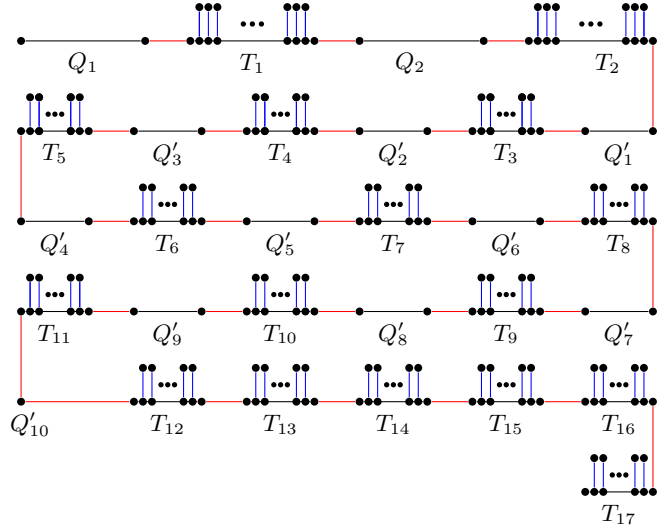

\begin{observation}\label{observation:overlap}
If $T^c_j$ is burnt by putting $m\geq 2$ fire sources on $T_j$, then the burning clusters of at least two of these fire sources overlap (i.e., contain common vertices) of $T_j$.
\end{observation}

\begin{proof}
Let that some $T^c_j$ be completely burnt by two or more fire sources and yet there is no overlap between the burning clusters of any of those fire sources.
Since all the fire sources are on $T_j$, which is a sub path of $P_I$, we say two fire sources on $T_j$ are \textit{adjacent} if there is a path in $T_j$ between those two fire sources such that the path does not contain any other fire sources.
For any two adjacent fire sources let us assume that there is no vertex which lies in the burning clusters of both the fire sources. Let $v$ be a vertex on the path joining those two adjacent fire sources $x_1$ and $x_2$, such that the vertices in the left side of $v$ including it (vertices towards $x_1$ as shown in Figure \ref{figure:TStructureIG} using the left arrow) are burnt by $x_1$ and the vertices in the right of $v$ (excluding $v$) are burnt by $x_2$.

Let the vertex that is just right to $v$ is $v_r$. By pigeonhole principle, we have that at least one of $v$ or $v_r$ having a neighbor $v_t$ in $T^c_j$ which is not in $T_j$. Without the loss of generality, let that $v$ is having such a neighbor. Since the burning cluster of $x_1$ extends till $v$ and not to its one hop neighbor $v_r\ (\in T_j)$, so it does not burn the other one hop neighbor $v_t\ (\not\in T_j)$ too. It is easy to see that the second fire source can not burn  $v_t$. This is contradiction to our assumption that $T^c_j$ is burnt completely without overlapping clusters.
\qed\end{proof}

The following observation is immediate.

\begin{observation}\label{observation:Tone}
If a single fire source is able to burn $T_j$ in $t$ rounds, then $T^c_j$ would also be burnt by it in the same number of rounds.
\end{observation}
\vspace{-0.3cm}
\begin{lemma}\label{lemma:notopt}
If at least one $T_j$ is burnt using more than one fire sources, then $P_I$ can not be burnt optimally, i.e., in $b(P_I)(= 2m+1)$ steps.
\end{lemma}
\vspace{-0.3cm}
\begin{proof}
Since $P_I$ is a simple path of length $(2m+1)^2$, according to \Cref{observation:path}, each fire source in an optimal burning sequence must burn disjoint set of vertices of $P_I$. Let ${x_1, x_2,...,x_{2m+1}}$ be an optimal burning sequence of $P_I$ such that some $T_j$ is burnt using more than one fire sources, then according to \Cref{observation:overlap}, at least two fire sources burn at least one common vertex of $P_I$ and hence ${x_1, x_2,...,x_{2m+1}}$ can not be an optimal burning sequence for $P_I$.
\qed\end{proof}

\subsection{NP-Completeness}
\label{sec:IGoptseq}
\begin{lemma}\label{lemma:BNb(G)IG}
If $X^\prime$ has a solution for the  distinct 3-partition problem, then burning number of $IG(X)$ is $2m+1$.
\end{lemma}
\vspace{-0.3cm}
\begin{proof}
If $X^\prime$ has a solution for the  distinct 3-partition problem, there would be $n$ sets of three numbers each, sum of which is $B'$. Recall that length of each $Q_i$ is $B'$. Hence, $Q_1,...,Q_n$ can be partitioned into further subpaths $Q^{\prime\prime}_1,...,Q^{\prime\prime}_{3n}$. Let us call the partitions of $Q_i$ as $Q^{\prime\prime}_{3(i-1)+1}$, $Q^{\prime\prime}_{3(i-1)+2}$, and $Q^{\prime\prime}_{3i}$. Since $X'$ is a set of odd numbers, length of each of these $3n$ subpaths are odd.

Let $P^\prime = \{Q^{\prime\prime}_1$, $...$, $Q^{\prime\prime}_{3n}$, $Q^\prime_1$, $...$, $Q^\prime_k$, $T_1$, $...$, $T_{m+1}\}$.
Let $r_i$ be the $((2m+1)-i+1)^{th} = (2m-i+2)^{th}$ vertex on the $i^{th}$ largest subpath in $P^\prime$. Then, the burning sequence $S^\prime = (r_1,r_2,..,r_{2m+1})$ can burn $P_I$ and subsequently $IG(X)$. This implies that $b(IG(X))\leq 2m+1$. Since $IG(X)$ has a subpath of length $(2m+1)^2$ in form of $P_I$, we have $b(IG(X))\geq 2m+1$. Hence, $b(IG(X)) = 2m+1$.
\qed\end{proof}
\begin{lemma}\label{lemma:EFSOnPath}
Each fire source $y_i$ of any optimal burning sequence $(y_1,y_2,...,y_{2m+1})$ of $IG(X)$ must be on $P_I$.
\end{lemma}
\vspace{-0.3cm}
\begin{proof}
We prove it by contradiction.
If for any $i$, $y_i$ is on $P_I$, then subgraph induced by  $G.N_{2m+1-i}[y_i]\cap P_I$ has length at most $2(2m+1-i)+1$.
Let we put a fire source $y_i$ on some vertex of $A^{T_j}$ for some $j$, which is not on $P_I$, and still burn $IG(X)$ in $2m+1$ steps. Then subgraph induced by $G.N_{2m+1-i}[y_i]\cap P_I$ is a path of length less than $2(2m+1-i)+1$. This along with \Cref{equation:WGBC} implies that  $|\cup_{i=1}^{2m+1}G.N_{2m+1-i}[y_i]\cap P_I|<(2m+1)^2$. So, even $P_I$ is not burnt. This is a contradiction to our assumption that $IG(X)$ is burnt in $2m+1$ steps. Therefore each $y_i$ must be a put on some vertex in $P_I$.
\qed\end{proof}

Let $S^\prime=(y_1,y_2,...,y_{2m+1})$ be any optimal \textit{burning sequence}. Let $r_i$ be the $(2m-i+2)^{th}$ vertex on the $i^{th}$ largest sub path in $P^\prime$ as described in the proof of Lemma \ref{lemma:BNb(G)IG}. Observe that $T_j$'s are the largest $m+1$ sub paths in $P^\prime$.
\begin{lemma}\label{lemma:AFSOnr_iIG}
We must have $y_i=r_i$, $\forall\ 1\leq i\leq m+1$.
\end{lemma}
\vspace{-0.3cm}
\begin{proof}
We are going to prove this lemma using the strong induction hypothesis. We have that each $u^k_j\in A^{T_j}$ for some $j$ must receive fire from some $y_i$ in $P_I$, as all fire sources must be on $P_I$ (Lemma \ref{lemma:EFSOnPath}).
For $i=1$, the only vertex connected to both $u^1_1$ and $u^{|T_1|-2}_1$ and within a distance $2m+1-i=2m$, is $r_1$.
Now we must have $y_1=r_1$, else, if we put $y_1$ somewhere else, then neither $y_1$ nor any other fire source can burn $T^c_1$ alone.
Also, we can not use multiple fire sources to burn $T^c_1$ as an optimal burning of $IG(X)$ does not allow that (Lemma \ref{lemma:notopt}).
So, we must have that $y_1=r_1$. Now to establish strong induction, let that we need to have $y_k$ on $r_k$ for $1\leq k\le m$.
Since $r_k$ is already used to burn $T^c_k$, the only fire source that can burn $T^c_{k+1}$ alone, is $r_{k+1}$ . Recall that $T_{k+1}$ has the largest length among the remaining subpaths after $T_1, T_2, \cdots, T_k$ are burnt. And we can not use multiple fire sources to burn $T^c_{k+1}$ (Lemma \ref{lemma:notopt}).
Also the only vertex connected to both $u^1_{k+1}$ and $u^{|T_{k+1}|-2}_{k+1}$  within distance $2m+1-(k+1)$ is $r_{k+1}$.
So, we must have that $y_{k+1}=r_{k+1}$. This completes the proof.
\qed\end{proof}

Let $P^{\prime\prime} = IG(X)\setminus(T^c_1\cup T^c_2\cup ... \cup T^c_{m+1})$. Now we present the following lemma on burning this remaining subgraph $P^{\prime\prime}$. That is $P^{\prime\prime}$ is a path forest consists of the subpaths $Q_1,Q_2...,Q_n,Q_1^\prime,Q_2^\prime,...,Q_k^\prime$. Now we present the following lemma on burning $P^{\prime\prime}$.
\begin{lemma}\label{lemma:PartPFx1TO31}
There is a partition of $P^{\prime\prime}$, induced by the fire sources $y_i$ $(m+1\leq i\leq 2m+1)$ of the optimal burning sequence $S^\prime$, into paths of orders in $F_m^\prime$.
\end{lemma}
\vspace{-0.3cm}
\begin{proof}
From Lemma \ref{lemma:AFSOnr_iIG}, we have that $\forall\ 1\leq i\leq m+1$, all the vertices in $T^c_i$, would be burnt by $y_i$. Therefore, we have to burn the vertices in $Q_1,...,Q_n,Q_1^\prime,...,Q_k^\prime$ by the fire sources $\{y_{m+2}$,$y_{m+3}$,$\cdots$, $y_{2m+1}\}$ (the remaining $m$ sources of fire). Since $P^{\prime\prime}$ is a disjoint union of paths, so we have that $\forall i$ such that $\ m+2\leq i\leq 2m+1$, the subgraph induced by the vertices in $G.N_{2m+1-i}[y_i]$ is a path of length at most $2(2m+1-i)+1$. Moreover, we have that the path forest $P^{\prime\prime}$ is of order $\sum_{i=1}^{m}(2i-1)=m^2$. This implies that for each $i$ with $\ m+2\leq i\leq 2m+1$, the subgraph induced by the vertices in $G.N_{2m+1-i}[y_i]$ is a path of order equal to $2(2m+1-i)+1$, otherwise we cannot burn all the vertices of $P^{\prime\prime}$ by these $m$ fire sources which is a contradiction to the fact that $S^\prime$ is a optimal burning sequence. Therefore there must be a partition of $P^{\prime\prime}$, induced by the burning sequence $y_{m+2},y_{m+3},...,y_{2m+1}$, into subpaths of length as per each element in $F_m^\prime = \{1, 3, 5, . . ., 2m-1\}$.
\qed\end{proof}
\vspace{-0.3cm}
\begin{theorem}\label{theorem:BIGNPCIG}
Optimal burning of an interval graph is NP-Complete.
\end{theorem}
\vspace{-0.3cm}
\begin{proof}
one part is already proved in Lemma \ref{lemma:BNb(G)IG}. Here we show the other part. Let say we have a $2m+1$ round optimal solution of the burning problem. 
Each $T^c_j$ must get burned by exactly one fire source as per Lemma \ref{lemma:AFSOnr_iIG}. From Lemma \ref{lemma:PartPFx1TO31}, we claim that the remaining path forest must be burned by the rest of available fire sources corresponding to the set $F'_m$.

Now, if each of the  $Q'_i$ is burned by a single fire source, then they must be burned by the fire sources corresponding to the integers belonging to set $Y$. Hence the remaining fire sources burning $Q_i$'s are burnt by fire sources corresponding to the integers belonging to the set $X'$. As the size of each $Q_i$ is $B'$, which is always odd ($B'=2B-3$), and also $B'> 2m -1$ (from the definition of distinct 3-partition problem, $m< B/2$), so no $Q_i$ can be burnt by a single fire source. Again it can not be burnt by two fire sources as sum of any two numbers in $X'$ are even. Also  no $Q_i$ can be burnt by 4 or more fire sources as then by pigeon whole principle there would be at least one $Q_i$ which needs to be burned by at most 2 fire sources, which is not possible. Hence each $Q_i$ must be burnt by exactly three fire sources.

Else, if $Q'_i$'s are not burned by single fire sources, we apply the following process subject to each subpath $Q^\prime_j$ for $1\leq j \leq k$. Let that some subpath $Q^\prime_j$ is burned using multiple fire sources. Since the given solution is optimal so these burning clusters are non overlapping. Not only that, the sum of the cluster sizes of these fire sources is exactly same as order of $Q^\prime_j$. Now some fire source with cluster size equal to order of $Q^\prime_j$ must be present on some other subpath. We can interchange that fire source (whose cluster size is $Q^\prime_j$) by these fire sources (which are presently burning $Q^\prime_j$). This way we can make each subpath $Q^\prime_j$ to be burnt by a single fire source whose cluster size is equal to $Q^\prime_j$. This takes $O(m)$ time. After this we again arrived to the case discussed above and we can see that each $Q_i$ are burnt by exactly three fire sources corresponding to the integers in $X'$. Therefore we have a solution of the distinct 3-partition problem whose input set is $X'$. This, in turn, gives us the solution of the distinct 3-partition problem on the input set $X$.

Therefore, we have reduced the burning problem of $IG(X)$ from the distinct 3-partition problem in pseudo-polynomial time. Since, the distinct 3-partition problem is NP-Complete in the strong sense, burning $IG(X)$ is also NP-Complete in the strong sense.
\qed\end{proof}

\section{Burning permutation graphs}\label{section:burn-pg}

First we define permutation graphs. A \textit{permutation graph} is constructed from an original sequence of objects $O = (1, 2, 3, . . ., k)$ which are numbers here and its permutation $P = (p_1, p_2, p_3, . . ., p_k)$ such that there is an edge between two vertices corresponding to number $i$ and $j$ respectively, if $i < j$ and $j$ occurs before $i$ in $P$.

We show in the following part of this section that burning permutation graphs is NP-Complete as well. The idea is similar to that of the interval graph. We start with an input $X$ of a  distinct 3-partition problem and reduce it to input $X'$ of another  distinct 3-partition problem. From that we construct a set of numbers $O$ and finally a permutation $P$ such that the permutation graph of $O$ w.r.t. $P$ becomes a path forest $P(X)$. This $P=$ $P_1$ $\cup_{s/}$ $P_2$ $\cup_{s/}$ $\cdots$ $\cup_{s/}$ $P_{n+k}$ is a sequential union of permutations $P_j$, such that the graph corresponds to $O$ and each $P_j$ forms a path. The path forest corresponding to $P$ is exactly similar to $P^{\prime\prime}$ that we used in the Lemma \ref{lemma:PartPFx1TO31}. More specifically, order wise path corresponding to each $P_j$ for all $1\le j \le n$ is same as the respective subpath $Q_i$ for all $1\le i \le n$ and path corresponding to  each $P_j$ for all $n+1\le j \le n+k$ is same as the respective subpath $Q'_i$ for all $1\le i\le k$ in $P^{\prime\prime}$. Then similar argument works here, i.e., if $X'$ has a solution for the  distinct 3-partition problem (and so is $X$), then we need to find that to burn the path forest corresponding to $P$,i.e., the permutation graph $P(X)$ optimally. So the main focus of this section is on the construction of such permutation graphs from given inputs of distinct 3-partition problem.

\subsection{Generation of arbitrary path forest from a permutation of numbers}

Let $X$ be an arbitrary multiset of $l=|X|$ positive integers. Let $O$ be the original sequence of numbers $1$ to $s(X)$.
Now, we are going to construct $|X|$ permutations $P_1, P_2,..., P_{l}$ in a specific manner such that these will produce path forest of $l$ disjoint simple paths. The permutation $P$ is simply the sequential union of the above permutations, i.e., $P=$ $P_1$ $\cup_{s/}$ $P_2$ $\cup_{s/}$ $\cdots$ $\cup_{s/}$ $P_{l}$.

Each $P_j =\{p^1_j,p^2_j,...,p^{t_j}_j\}$ is a permutation of the consecutive numbers $x_j$ to $y_j$ belonging to $O$ where $t_j=y_j-x_j+1$. The sets of such consecutive numbers those corresponds to the $l$ permutations, are pairwise disjoint. This makes the union of all paths corresponding to the permutations, a path forest. 

We provide a formula to compute $x_j$, $y_j$ as follows. Let $y_0=0$. $\forall\ 1\leq j\leq l$,
$x_j = y_{j-1}+1$ and  $y_j = y_{j-1} + X[j]$, where $X[j]$ is the $j^{th}$ element of $X$. This also implies that the total number of vertices in the permutation graph shall be $y_{l} = s(X)$.

We provide general formula to find $p^h_j$ for each $j$ and for all $h \in (1, t_j)$ such that
$P_j$ corresponds to a path.
\begin{itemize}
\item {\bf {For all}\boldmath {$\ 1\leq j \leq l$}\bf{ s.t.}\boldmath{ $t_j$}\bf{ is even with}\boldmath{ $|t_j|\ge 5$:}}
\subitem For all odd $h,\ 1 \leq h \leq (t_j-3), p^h_j = 2+(x_j+h-1)$. For the remaining odd value of $h=t_j-1$, $p^{t_j-1}_j = y_j$.
\subitem For all even $h, 4 \leq i \leq t_j, p^h_j = h-2$ and for the remaining even value of $h=2$, $p^2_j = x_j$\\

\item {\bf {For all} \boldmath {$\ 1\leq j \leq l$}\bf{ s.t.}\boldmath{ $t_j$}\bf{ is odd with}\boldmath{ $|t_j|\ge 5$:}}
\subitem For all odd $h,\ 1 \leq h \leq t_j-2, p^h_j = 2+(x_j+h-1)$. For the remaining odd value of  $h=t_j$,  $p^{t_j}_j = y_j-1$.
\subitem For all even $h, 4 \leq h \leq t_j-1, p^h_j = (x_j+h-1)-2$ and for the remaining even value of $h=2$, $p^2_j = x_j$\\

\item {\boldmath {$|t_j|\le 4$:}}
\subitem If $t_j=1$, then $P_j=(x_j)$.
\subitem If $t_j=2$, then $P_j=(y_j,x_j)$.
\subitem If $t_j=3$, then $P_j=(y_j,x_j,x_j+1)$.
\subitem If $t_j=4$, then $P_j=(x_j+1,y_j,x_j,x_j+2)$
\end{itemize}

We also provide an example of such a construction that returns a path.
Let $O = (1, 2, 3, 4, 5, 6, 7, 8)$. Then consider $P = (3, 1, 5, 2, 7, 4, 8, 6)$ be the subject permutation of $O$. The permutation graph $G$ formed from this pair $(O, P)$ is shown in Figure \ref{figure:example-permutation-graph}. In the figure, $v_i$ are the vertices corresponding to the object $i\in O$.

\begin{figure}[ht!]
	\centering
	\begin{tikzpicture}
	    \draw (0,0) -- (7,0);
	
	    \node [circle, fill=black, inner sep=0pt, minimum size=3pt, label=below:{$v_1$}] at (0,0) {};
		\node [circle, fill=black, inner sep=0pt, minimum size=3pt, label=below:{$v_3$}] at (1,0) {};
	    \node [circle, fill=black, inner sep=0pt, minimum size=3pt, label=below:{$v_2$}] at (2,0) {};
		\node [circle, fill=black, inner sep=0pt, minimum size=3pt, label=below:{$v_5$}] at (3,0) {};
	    \node [circle, fill=black, inner sep=0pt, minimum size=3pt, label=below:{$v_4$}] at (4,0) {};
	    \node [circle, fill=black, inner sep=0pt, minimum size=3pt, label=below:{$v_7$}] at (5,0) {};
	    \node [circle, fill=black, inner sep=0pt, minimum size=3pt, label=below:{$v_6$}] at (6,0) {};
		\node [circle, fill=black, inner sep=0pt, minimum size=3pt, label=below:{$v_8$}] at (7,0) {};
	\end{tikzpicture}
	\caption{Representation of permutation graph corresponding to $(O, P)$, where $O = (1, 2, 3, 4, 5, 6, 7, 8)$ and $P = (3, 1, 5, 2, 7, 4, 8, 6)$.}
	\label{figure:example-permutation-graph}
\end{figure}
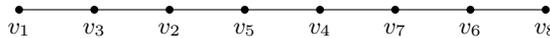

Now, $P$ $=$ $P_1$ $\cup_{s/}$ $...$ $\cup_{s/}$ $P_{l}$ $=$ $(p^1_1$, $...$, $p^{t_1}_1$, $p^1_2$, $...$, $p^{t_2}_2$, $...$, $p^1_{l}$, $...$, $p^{t_{l}}_{l})$ is the subject permutation of $O$.

We call $P(X)$ to be the permutation graph corresponding to the original sequence $O$, and its subject permutation $P$. $\forall\ 1\leq j\leq l$ let $Q_j$ be the subgraph in $P(X)$ induced by the permutation $P_j=(p^1_j,p^2_j,...,p^{t_j}_j)$ of the original sequence $(x_j,...,y_j)$. Observe that $P(X) = Q_1\cup Q_2\cup...\cup Q_{n+k}$ is a path forest where the paths $Q_1, Q_2,..., Q_{n+k}$ are disjoint from each other.

Although it immediately follows from \cite{Bessy2017} that burning permutation graphs is NP-Complete since we can construct any path forest from permutation graphs, we are giving a reduction of the distinct 3-partition problem in the following part of this section to show NP-Completeness.

\subsection{Permutation graph general construction for NP-Completeness}

Let $X$ be an input set to a  distinct 3-partition problem; let $n=\frac{|X|}{3}$, $m = \max (X)$, $B = \frac{s(X)}{n}$, and $k=m-3n$. Let $F_m$ be the set of first $m$ numbers, $F_m = \{1,2,3,...,m\}$, and $F^{\prime}_m$ be the set of first $m$ odd numbers, $F^\prime_m = \{2\ f_i-1: f_i \in F_m\} = \{1, 3, 5, . . ., 2m-1\}$. Let $X^\prime = \{2\ a_i-1:a_i \in X\}$, $B^\prime = \frac{s(X^\prime)}{n}$. Observe that $s(X^\prime) = \sum_{i=1}^{3n} 2\ a_i -1 = 2nB-3n$, so $B^\prime = 2B-3$. Let $Y=F^\prime_m\setminus X^\prime$. Let $O$ be the original sequence of numbers $1$ to $s(F^\prime_m)$, $O=(1,2,3,...,m^2)$.
Now, we are going to construct $n+k$ permutations $P_1, P_2,..., P_{n+k}$ in a specific manner such that these will produce path forest of $(n+k)$ disjoint simple paths, where $k=|Y|$. The permutation $P$ is simply the sequential union of the above permutations, i.e., $P=$ $P_1$ $\cup_{s/}$ $P_2$ $\cup_{s/}$ $\cdots$ $\cup_{s/}$ $P_{n+k}$.

Each $P_j =\{p^1_j,p^2_j,...,p^{t_j}_j\}$ is a permutation of the consecutive numbers $x_j$ to $y_j$ belonging to $O$ where $t_j=y_j-x_j+1$. The sets of such consecutive numbers those corresponds to the $(n+k)$ permutations, are pairwise disjoint. This makes the union of all paths corresponding to the permutations, a path forest. 
Below we first provide a formula to calculate $x_j$, $y_j$.
Since each of $P_1, P_2, \cdots, P_n$ is of order $B'$ (i.e., $t_j=B^\prime$ $\forall 1\le j \le n$), hence the $P_j$ is a permutation of the numbers $(j-1)\times B^\prime+1$ to $j\times B^\prime$. We define this formally in the following way. Let $y_0=0$. Then $\forall\ 1\leq j\leq n$, $x_j = y_{j-1}+1,$ and $y_j = j \times B^\prime$.
For the remaining permutations  $P_{n+1}, P_{n+2}, \cdots, P_{n+k}$, $\forall\ n+1\leq j\leq n+k$,
$x_j = y_{j-1}+1$ and  $y_j = y_{j-1} + L^Y_i$, where $i=j-n$ and $L^Y_i$ is the $i^{th}$ largest element of $Y$. Note that, $i$ varies from 1 to $k$.
Also observe the following,
\begin{align*}
y_{n+k} &= y_n + s(F^\prime_m\setminus X^\prime) = nB^\prime + s(F^\prime_m\setminus X^\prime)\\
&= s(X^\prime) + s(F^\prime_m\setminus X^\prime) = s(F^\prime_m) = m^2
\end{align*}
Hence, total number of elements in $\bigcup^{n+k}_{j=1} P_j$ is $m^2$.

Now, $P$ $=$ $P_1$ $\cup_{s/}$ $...$ $\cup_{s/}$ $P_{n+k}$ $=$ $(p^1_1$, $...$, $p^{t_1}_1$, $p^1_2$, $...$, $p^{t_2}_2$, $...$, $p^1_{n+k}$, $...$, $p^{t_{n+k}}_{n+k})$ $=$ $(p_1$, $p_2$, $...$, $p_{m^2})$ is the subject permutation of $O$.

We call $P(X)$ to be the permutation graph corresponding to the original sequence $O$, and its subject permutation $P$. $\forall\ 1\leq j\leq n+k$ let $Q_j$ be the subgraph in $P(X)$ induced by the permutation $P_j=(p^1_j,p^2_j,...,p^{t_j}_j)$ of the original sequence $(x_j,...,y_j)$. Observe that $P(X) = Q_1\cup Q_2\cup...\cup Q_{n+k}$ is a path forest where the paths $Q_1, Q_2,..., Q_{n+k}$ are disjoint from each other.
Now we prove the NP-Completeness result which mostly follows from the earlier proof on the interval graphs.

\subsection{Example Construction}

Let $X=\{10,11,12,14,15,16\} \implies n=2,\ m = 16,\ B = 39,$ and $k=10$. $F_m=\{1,2,...,16\}$, and $F_m^\prime=\{1,3,...,31\}$. $X^\prime = \{19,21,23,27,29,31\}$, $B^\prime = 75=2B-3$. $Y=\{1,3,5,7,9,11,13,15,17,25\}$.

We finally form paths $Q_1$ and $Q_2$ each of order of $75$. Also, we form paths $Q_3,Q_4,...,Q_{12}$ of order of $25,17,15,13,11,9,7,5,3,1$ respectively. $P(X)$ is a path forest of the paths $Q_1,...,Q_{12}$, which are disjoint from each other. Burning number of $P(X)$ in this case is $m=16$.

\subsection{NP-Completeness}

The logic is similar to that of the interval graph.
The burning number of $P(X)$ is $m$ if $X'$ has a solution of the  distinct 3-partition problem. As then length $B'$ of each of $Q_1, Q_2,...Q_n$ can be written as a sum of three odd numbers from $X'$. And length of each of the remaining $Q_{n+1}, Q_{n+2},..., Q_{n+k}$ are also odds as each of the lengths corresponds to an element in $Y$. So, to burn $P(X)$ in $m$ steps, one needs to solve the  distinct 3-partition problem on $X'$ (equivalently on $X$).

Similarly if we have an optimal burning sequence of $P(X)$, we get a solution of the distinct 3-partition problem for $X$ as discussed the last paragraph of the proof of Theorem \ref{theorem:BIGNPCIG}. This leads us to the following theorem.

\begin{theorem}\label{theorem:NPCPG}
    Burning of general permutation graphs is NP-Complete.
\end{theorem}

\section{Discussion}\label{section:conclude}

\subsection{Some more Hardness Results}\label{sec:hardgeometric}

In this section we report hardness results on few more graph classes that mostly follow from our result on the interval graph.
A \textit{disk graph} is formed from an arrangement of disks on a Euclidean plane such that there is a vertex in the disk graph corresponding to each disk, and if there is an overlap between a pair of disks, then there shall be an edge between their corresponding vertices in the disk graph. Since any interval graph is valid to be a disc graph, we have the following.
\vspace{-0.2cm}
\begin{corollary}
    Optimal burning of disc graphs is NP-Complete even if the underlying disc representation is given.
\end{corollary}

In a \textit{unit distance graph}, the edges can be drawn in a euclidean plane such that each edge is of unit length. In \textit{matchstick graph}, the edges can be drawn in a euclidean plane such that each edge is of unit length and they do not intersect each other. The graph class that we have constructed is valid to be a unit distance graph and a matchstick graph. So we have \Cref{corollary:unit-distance-matchstick} as follows.

\begin{corollary}\label{corollary:unit-distance-matchstick}
    Optimal burning of unit distance graphs and matchstick graphs is NP-Complete.
\end{corollary}

\subsection{Conclusion and future work}

In this article we show a lower bound for the burning number of grids of arbitrary size and give a two approximation algorithm for burning square grids. We also show that the graph burning problem is NP-Complete on interval graphs  and permutation graphs along with several corollaries that show NP-Completeness of burning other geometric graphs as well.

It remains an open question whether burning grids is an NP-Complete problem. Another related direction is to try and improve the 3-approximation algorithm provided in \cite{Bessy2017} for burning general graphs.

\bibliography{ref.bib}

\begin{thebibliography}{10}

\bibitem{Simon2019}
M.~\~Simon, L.~Huraj, I.~Dirgov{\~a}~Lupt{\~a}kov{\~a}, and J.~Posp{\'i}chal.
\newblock Heuristics for spreading alarm throughout a network.
\newblock {\em Applied Sciences}, 9(16), 2019.

\bibitem{Alon2009}
N.. Alon, P.. PraLat, and N.. Wormald.
\newblock Cleaning regular graphs with brushes.
\newblock {\em SIAM Journal on Discrete Mathematics}, 23(1):233--250, 2009.

\bibitem{Balogh2012}
J.~Balogh, B.~Bollob{\'a}s, and R.~Morris.
\newblock Graph bootstrap percolation.
\newblock {\em Random Structures and Algorithms}, 41(4):413--440, 12 2012.

\bibitem{Banerjee2012}
S.~Banerjee, A.~Gopalan, A.~Das, and S.~Shakkottai.
\newblock Epidemic spreading with external agents.
\newblock {\em IEEE Transactions on Information Theory}, 60, 06 2012.

\bibitem{Bessy2017}
S.~Bessy, A.~Bonato, J.~Janssen, D.~Rautenbach, and E.~Roshanbin.
\newblock Burning a graph is hard.
\newblock {\em Discrete Appl. Math.}, 232(C):73--87, 2017.

\bibitem{Bessy2018}
S.~Bessy, A.~Bonato, J.~Janssen, D.~Rautenbach, and E.~Roshanbin.
\newblock Bounds on the burning number.
\newblock {\em Discrete Applied Mathematics}, 235:16 -- 22, 2018.

\bibitem{Bonato2016}
A.~Bonato, J.~Janssen, and E.~Roshanbin.
\newblock How to burn a graph.
\newblock {\em Internet Mathematics}, 12(1-2):85--100, 2016.

\bibitem{Bonato2019}
A.~Bonato and S.~Kamali.
\newblock Approximation {Algorithms} for {Graph} {Burning}.
\newblock In T.V. Gopal and Junzo Watada, editors, {\em Theory and
  {Applications} of {Models} of {Computation}}, pages 74--92, Cham, 2019.

\bibitem{Bonato2019a}
A.~Bonato and T.~Lidbetter.
\newblock Bounds on the burning numbers of spiders and path-forests.
\newblock {\em Theoretical Computer Science}, 794:12 -- 19, 2019.

\bibitem{sandip2018}
S.~Das, S.~Ranjan Dev, A.~Sadhukhan, U.~Kant Sahoo, and S.~Sen.
\newblock Burning spiders.
\newblock In B.~S. Panda and Partha~P. Goswami, editors, {\em CALDAM}, volume
  10743 of {\em Lecture Notes in Computer Science}, pages 155--163. Springer,
  2018.

\bibitem{Finbow2009}
S.~Finbow and G.~Macgillivray.
\newblock The firefighter problem: A survey of results, directions and
  questions.
\newblock {\em The Australasian Journal of Combinatorics [electronic only]},
  43, 02 2009.

\bibitem{Garey1979}
M.~R. Garey and D.~S. Johnson.
\newblock {\em Computers and Intractability: A Guide to the Theory of
  NP-Completeness}.
\newblock W. H. Freeman \& Co., USA, 1979.

\bibitem{Kamali2020}
S.~Kamali, A.~Miller, and K.~Zhang.
\newblock Burning two worlds.
\newblock In {\em SOFSEM}, 2020.

\bibitem{Kare2019}
A.~S. Kare and I.~V.~Reddy.
\newblock Parameterized {Algorithms} for {Graph} {Burning} {Problem}.
\newblock In {\em Combinatorial {Algorithms}}, pages 304--314, Cham, 2019.

\bibitem{Kempe2003}
D.~Kempe, J.~Kleinberg, and \'{E}. Tardos.
\newblock Maximizing the spread of influence through a social network.
\newblock In {\em ACM SIGKDD}, pages 137--146, New York, NY, USA, 2003.
  Association for Computing Machinery.

\bibitem{Kempe2005}
D.~Kempe, J.~Kleinberg, and \'{E}. Tardos.
\newblock Influential nodes in a diffusion model for social networks.
\newblock In {\em ICALP}, pages 1127--1138, Berlin, Heidelberg, 2005.
  Springer-Verlag.

\bibitem{Kramer2014}
A.~D.~I. Kramer, J.~E. Guillory, and J.~T. Hancock.
\newblock Experimental evidence of massive-scale emotional contagion through
  social networks.
\newblock {\em Proceedings of the National Academy of Sciences},
  111(24):8788--8790, 2014.

\bibitem{Simon2019a}
M.~Simon, L.~Huraj, Dirgova L., and J.~Pospichal.
\newblock How to burn a network or spread alarm.
\newblock {\em MENDEL}, 25(2):11--18, Dec. 2019.

\end{thebibliography}
\bibliographystyle{plain}

\end{document}